\documentclass[11pt]{article}
\usepackage[utf8]{inputenc}
\usepackage[left=1in, right=1in, bottom=1.15in, top=1.1in]{geometry}

\usepackage[dvipsnames]{xcolor}

\usepackage{amsmath,amssymb}
\usepackage{mathtools}
\usepackage{hyperref}
\hypersetup{%
  colorlinks=true,
  citecolor=MidnightBlue,
  linkcolor=MidnightBlue,
  urlcolor=MidnightBlue,
}
\usepackage{url}
\usepackage{hyperref}
\usepackage{bbm}
\usepackage{cleveref}

\usepackage[vlined,ruled,linesnumbered]{algorithm2e} 
\SetKwRepeat{Do}{do}{while}

\usepackage{verbatim} 
\usepackage{multirow}
\usepackage{amsthm} 
\newtheorem{theorem}{Theorem}
\newtheorem{corollary}{Corollary}
\newtheorem{lemma}{Lemma}

\newtheorem{definition}{Definition}

\newtheorem{axiom}{Axiom}
\newtheorem{remark}{Remark}

\usepackage{threeparttable}
\usepackage{graphicx}
\usepackage{subfigure}
\usepackage{tikz}
\usetikzlibrary{automata, positioning, arrows}
\tikzset{
box/.style ={
rectangle, 
rounded corners =1pt, 
minimum width =50pt, 
minimum height =20pt, 
inner sep=5pt, 
draw=black, 
align = center
}}

\newcommand{\calD}{\mathcal{D}}

\newcommand{\calX}{\mathcal{X}}
\newcommand{\calY}{\mathcal{Y}}

\def\TX{\textsf{TX}}

\def\Sell{\texttt{Sell}}
\def\GSR{\texttt{GSR}}
\def\AP{\texttt{AP}}

\title{MEV Makes Everyone Happy under Greedy Sequencing Rule}

\author{
Yuhao Li\footnote{Equal contribution.}
\\\textit{Columbia University}\\ \texttt{yuhaoli@cs.columbia.edu} \hspace{-0.42cm}
\and Mengqian Zhang$^*$ \\\textit{New York University} \\
\texttt{mengqian.zhang@stern.nyu.edu}\hspace{-0.42cm}
\and Jichen Li$^*$ \\\textit{Peking University}\\ \texttt{limo923@pku.edu.cn}
\vspace{0.2cm}
\and Elynn Chen \\\textit{New York University}\\ \texttt{elynn.chen@stern.nyu.edu}
\and Xi Chen \\\textit{New York University}\\ \texttt{xc13@stern.nyu.edu}
\and Xiaotie Deng \\\textit{Peking University}\\ \texttt{xiaotie@pku.edu.cn}
}

\date{}

\begin{document}

\maketitle
\begin{abstract}
Trading through decentralized exchanges (DEXs) has become crucial in today's blockchain ecosystem, enabling users to swap tokens efficiently and automatically.
However, the capacity of miners to strategically order transactions has led to exploitative practices (\textit{e.g.}, front-running attacks, sandwich attacks) and gain substantial Maximal Extractable Value (MEV) for their own advantage. To mitigate such manipulation, Ferreira and Parkes recently proposed a greedy sequencing rule such that the execution price of transactions in a block moves back and forth around the starting price. Utilizing this sequencing rule makes it impossible for miners to conduct sandwich attacks, consequently mitigating the MEV problem.  

However, no sequencing rule can prevent miners from obtaining risk-free profits. This paper systemically studies the computation of a miner's optimal strategy for maximizing MEV under the greedy sequencing rule, where the utility of miners is measured by the overall value of their token holdings. Our results unveil a dichotomy between the no trading fee scenario, which can be optimally strategized in polynomial time, and the scenario with a constant fraction of trading fee, where finding the optimal strategy is proven NP-hard. The latter represents a significant challenge for miners seeking optimal MEV.

Following the computation results, we further show a remarkable phenomenon: Miner's optimal MEV also benefits users. Precisely, in the scenarios without trading fees, when miners adopt the optimal strategy given by our algorithm, all users' transactions will be executed, and each user will receive equivalent or surpass profits compared to their expectations. This outcome provides further support for the study and design of sequencing rules in decentralized exchanges.

\vspace{0.5cm}
\end{abstract}

\thispagestyle{empty}
\newpage 
\setcounter{page}{1}

\section{Introduction}
Decentralized finance (also known as DeFi), as the main application of blockchain and smart contracts, has grown incredibly popular and attracted more than 40 billion dollars~\cite{defillama}. Within the DeFi ecosystem, decentralized exchange (DEX) becomes a fundamental service that allows users to trade cryptocurrency directly without any centralized authority. Nowadays, the daily volume of these DEXs has reached billions of US dollars~\cite{defi-tracker}.

Most DEXs (\textit{e.g.}, Uniswap, SushiSwap, Curve Finance, and Balancer) are organized as constant function market makers (CFMMs). Uniswap~\cite{adams2021uniswap}, for example, utilizes a constant product formula to make sure that the product of the quantity of two tokens remains constant before and after a swap. The exchange rate, or say the price that the swap executes at, is automatically determined by the reserves of the pair. So the outcome of each trade is sensitively influenced by system status at execution time.

In the blockchain, it is the block builders (also referred to as miners or validators) that select pending transactions and specify their execution order. This gives an exploitable chance for miners to extract profit by strategically including, excluding, and reordering transactions in a block. This is known as \textit{Maximal Extractable Value} (MEV)~\cite{Daian20flash}. A prevalent MEV example is the \textit{sandwich attack}~\cite{zhou21high} on DEX transactions: the attacker ``sandwiches'' a profitable victim transaction by front-running and back-running it and earns from the spread between buying and selling prices.

To mitigate this market manipulation by miners, Ferreira and Parkes~\cite{xavier2023credible} recently introduced a \textit{greedy sequencing rule}. Simply put, when dealing with a bunch of transactions from the liquidity pool of tokens $\calX$ and $\calY$, this sequencing rule requires miners to take the starting price as a benchmark. Then at any point during the execution in the block, if the current price of $\calY$ is higher than the benchmark, the priority should be given to the transactions selling token $\calY$. Conversely, the transactions selling token $\calX$ should be executed next. This sequencing rule structurally makes the sandwich attack impossible. It restricts miners from manipulating transaction orders, thus mitigating the impact of MEV. More importantly, it introduces verifiability by allowing users to efficiently verify whether the execution order of transactions complies with the rule.

\subsection{Our Contributions}
As mentioned in~\cite{xavier2023credible}, miners can always obtain risk-free profits in some cases under arbitrary sequencing rule. In this paper, we systematically study the computation of miner's optimal MEV strategy under the greedy sequencing rule. The study is based on the utility model where the worth of miners is the overall \textit{value} of all their tokens. Like the similar work~\cite{bartoletti22maximizing} aiming to maximize extractable value without rules or limits, the value of a token is measured by its price, which is exogenous, given by an oracle, and fixed throughout the attack. It was explicitly emphasized by Ferreira and Parkes~\cite{xavier2023credible} to also consider miner's utility as a real-valued function when studying sequencing rules. The monetary function we considered is arguably the most natural choice.

We highlight our results on the computation of miners' optimal strategies, as well as their surprising consequences. We give a computation dichotomy, supported by our two main theorems (\Cref{theorem: polytime algorithm f=0} and \Cref{theorem: NP hard when f>0}). 
For the scenario where there is no trading fee, a polynomial time algorithm for a miner to compute an optimal strategy is given (\Cref{theorem: polytime algorithm f=0}); In contrast, when the fraction of trading fees is any constant larger than 0 (\textit{e.g.}, $f=0.3\%$ in most Uniswap pools), we prove it is NP-hard to find an optimal strategy (\Cref{theorem: NP hard when f>0}). 

The computational intractability implies hardness for a miner to hope for optimal MEV. More surprisingly, in the $f=0$ regime, when miners adopt the optimal strategy provided by our algorithm (Algorithm~\ref{alg: f=0}), users will also benefit in the following sense: all users' transactions will be executed (\Cref{corollary: liveness}), and every user gets at least as good as if their transaction was the only transaction in the block (\Cref{corollary: good profit}). The latter was one of the main motivations to propose the greedy sequencing rule, even though it is generally not true when the miner \emph{truthfully} follows the greedy sequencing rule.

We conclude this paper by discussing many interesting future directions and open problems in the last section (\Cref{sec: discussion}).

\subsection{Related Work}

\subsubsection{Sequencing Rules}
Typically, miners organize transactions based on their gas prices. In order to protect users from order manipulation, Kelkar \textit{et al.}~\cite{Kelkar0GJ20} investigate the notion of \textit{fair transaction ordering} for Byzantine consensus, which is further extended to the permissionless setting in~\cite{KelkarDK22Order-Fair}. Cachin \textit{et al.}~\cite{CachinMSZ22} introduce a new \textit{differential order-fairness property} and present the quick order-fair atomic broadcast protocol which is much more efficient than previous solutions. The general idea of these approaches is to rely on a committee rather than a single miner to order transactions. A main threat to fair transaction ordering is the \textit{Condorcet attack}~\cite{vafadar2023condorcet}. Vafadar and Khabbazian~\cite{vafadar2023condorcet} show that an attacker can undermine fairness by imposing
Condorcet cycles even when all others in the system behave honestly. 

Another category is \textit{content-oblivious ordering}~\cite{MalkhiS22, Sikka} which guarantees that the transaction data is not accessible to the committee responsible for sequencing them until an order has been determined. This could be achieved using methods like threshold public key encryption schemes.

\subsubsection{MEV Mitigation}
It has long been discovered that miners could exploit transaction ordering for their own benefit ~\cite{bracciali19transparent}. The term Maximal Extractable Value (MEV) was introduced in~\cite{Daian20flash}, formally defined in~\cite{BabelDKJ23}, and its growth has resulted in network congestion and high gas prices~\cite{Daian20flash, kulkarni22towards}. Besides the sequencing rules, some other approaches are also explored to mitigate the impact of MEV. To avoid sandwich attacks, users are suggested to reduce the trading volume by splitting transactions~\cite{zust2021analyzing} and to restrict the slippage tolerance~\cite{heimbach22eliminating}. This method, however, may also increase the transaction costs for users. Zhou \textit{et al.}~\cite{zhou21a2mm} propose a new DEX design called A$^2$MM, which helps users to immediately execute an arbitrage following their swap transactions. It also allows users to benefit from MEV atomically.
Another popular way is to rely on the service from trusted third parties like flashbots~\cite{flashbots}, Eden~\cite{Eden}, and OpenMEV~\cite{OpenMEV}. Then can help to order transactions without broadcasting them to the whole network, thus protecting from front-running and sandwich attacks.

\section{Preliminaries}
\subsection{Constant Function Market Makers}
Let $A$ be an AMM for trading between token $\calX$ and token $\calY$. The exchange has \emph{state} $s=(x,y)$, where $x$ and $y$ are current reserves of token $\calX$ and $\calY$, respectively.
When $A$ is a CFMM, the trading invariant can be modeled by a constant function with two variables $F(x,y)=C$. We will focus on CFMMs that satisfy Axiom~\ref{axiom: monotone} and Axiom~\ref{axiom: marginal exchange rate decreases}, which are defined as follows. We note that all currently known CFMMs are consistent with these two properties.\footnote{This also includes Uniswap v3, which is less trivial.}

\begin{axiom}\label{axiom: monotone}
    For different pairs $(x,y)$ and $(x',y')$ such that $F(x,y)=F(x',y')=C$, we have $x<x'$ if and only if $y>y'$.
\end{axiom}
By this axiom, we know that for any $x$ (reserves of token $\calX$), there is a unique $y$ such that $F(x,y)=C$ and vice versa. So we will use $F_y(x)$ to denote the $y$ such that $F(x,y)=C$ and similarly define $F_x(y)$.

\begin{axiom}\label{axiom: marginal exchange rate decreases}
    $F_y(x)$ is differentiable and the marginal exchange rate $|dF_y(x)/dx|$ is decreasing with respect to $x$.
\end{axiom}

In the rest of the paper, we use $r(x)$ to denote the \textit{marginal exchange rate} of swapping tokens $\calX$ for $\calY$, \textit{i.e.}, $r(x)\coloneqq |dF_y(x)/dx|$.

\subsection{Execution of Transactions}

Users can submit a transaction of the following two types: $\Sell(\calX,\cdot)$ and $\Sell(\calY,\cdot)$, where $\cdot$ is a real parameter representing how many units of token the user wants to trade. 

To be more concrete, suppose that the current state of CFMM $A$ is $s=(x,y)$. For each swap, part of tokens are charged as fees and we use $f\in[0,1)$ to denote the fraction of this trading fee. When executing a transaction $\Sell(\calX,q)$, the user will pay $q$ units of token $\calX$ and get $y-F_y(x+(1-f)q)$ units of token $\calY$. Similarly, when executing a transaction $\Sell(\calY, q)$, the user will pay $q$ units of token $\calY$ and get $x-F_x(y+(1-f)q)$ units of token $\calX$.

The executing of multiple transactions $\{\TX^i\}_{i\in[n]}$ will be well-defined if an order among them is determined. In particular, suppose that $\tau:[n]\rightarrow[n]$ is a permutation. Then the execution will work as follows: Let $s_0=(x_0,y_0)$ be the initial state and iteratively execute each transaction $\TX^{\tau(i)}$. For the $i$-th iteration, if $\TX^{\tau(i)}=\Sell(\calX,q)$, then $s_i=(x_i,y_i)$ where $x_i=x_{i-1}+(1-f)q$ and $y_i=F_y(x_i)$; if $\TX^{\tau(i)}=\Sell(\calY,q)$, then $s_i=(x_i,y_i)$ where $y_i=y_{i-1}+(1-f)q$ and $x_i=F_x(y_i)$.

It is easy to see the order under which the transactions are executed crucially influences the trades outcomes. However, due to the same reason, it is also well-known that the decentralized exchange systems suffer from \emph{order manipulation}, where an anonymous miner can manipulate the context of a block, even including inserting their own attacking transactions. Ferreira and Parkes~\cite{xavier2023credible} considered the notion of \emph{verifiable sequencing rules} and proposed a greedy sequencing rule to limit miners' ability to manipulate (therefore in general it also benefits users). We recap their definitions below.

\subsection{Sequencing Rules}
We start with the definition of the verifiable sequencing rule.

\begin{definition}[Verifiable sequencing rule, \cite{xavier2023credible}]
    A sequencing rule $R$ is a map from initial state $s_0$ and a set of transactions $\{\TX^i\}_{i\in[n]}$ to a set of permutations $\{\tau:[n]\rightarrow[n]\}$, where each permutation is a valid order to execute these transactions under this sequencing rule.
    
    A sequencing rule is \emph{efficiently computable}, if there is a polynomial time algorithm that can compute a permutation $\tau:[n]\rightarrow[n]$ that satisfies $R$ (\textit{i.e.}, $\tau\in R(s_0,\{\TX^i\}_{i\in[n]})$) for any initial state $s_0$ and transactions $\{\TX^i\}_{i\in[n]}$.
    
    A sequencing rule is \emph{efficiently verifiable}, if there is a polynomial time algorithm such that for any permutation $\tau:[n]\rightarrow[n]$, the algorithm accepts $\tau$ if and only if $\tau\in R(s_0,\{\TX^i\}_{i\in[n]})$.
\end{definition}

Along this way, Ferreira and Parkes~\cite{xavier2023credible} proposed a greedy sequencing rule (we use $\GSR$ to denote it), which is efficiently computable and verifiable.
\begin{definition}[Greedy sequencing rule, \cite{xavier2023credible}]
     A permutation $\tau$ satisfies the greedy sequencing rule ($\tau\in \GSR(s_0,\{\TX^{i}\}_{i\in[n]})$) if the following conditions hold for all $i\in[n]$:
     \begin{itemize}
         \item $\TX^{\tau(i)}$ is a $\Sell(\calX,\cdot)$ transaction only if \textit{either} $x_{i-1}\leq x_0$ \textit{or} $\TX^{\tau(j)}$ is $\Sell(\calX,\cdot)$ for all $i< j\leq n$; and
         \item $\TX^{\tau(i)}$ is a $\Sell(\calY,\cdot)$ transaction only if \textit{either} $y_{i-1}\leq y_0$ \textit{or} $\TX^{\tau(j)}$ is $\Sell(\calY,\cdot)$ for all $i< j\leq n$,
     \end{itemize}
     where $s_{i-1}=(x_{i-1},y_{i-1})$ is the state before executing $\TX^{\tau(i)}$.
\end{definition}

Besides efficiency, the greedy sequencing rule enjoys the property that for every transaction, \textit{either} its receive is as good as it was the only transaction in the block \textit{or} it does not suffer from a sandwich attack.

However, it is totally possible for a miner to gain profits by manipulating the content of the block, even if it follows some given sequencing rule (e.g., the greedy sequencing rule). In the rest of the paper, we study the computation of miners' optimal strategies.

\section{Miner's Strategy Space}

We define the miner's strategy space in the most general way. To make the profits of the miner comparable, we assume that there are exogenous prices of $\calX$ (denoted by $p_x$) and $\calY$ (denoted by $p_y$) and the miner wants to collect as much money as possible. Like previous work~\cite{bartoletti22maximizing}, $p_x$ and $p_y$ are assumed to remain the same during the attack (usually the timeslot for a block, \textit{e.g.}, about 12 seconds in Ethereum).

\begin{definition}[Strategy Space]
    Given a sequencing rule $R$, an initial state $s_0=(x_0,y_0)$, and a set of users' transactions $\{\TX^{i}\}_{i\in[n]}$, a miner could create $m$ number of its own transactions $\{\TX^i\}_{i\in[n+1:n+m]}$, select a subset of all these $n+m$ transactions $S\subseteq [n+m]$, compute an order $\tau\in R(s_0,\{\TX^i\}_{i\in S})$ (here instead of permutation, $\tau$ should be a one-to-one mapping from $[|S|]$ to $S$) that satisfies the sequencing rule, and execute them under the order $\tau$.

    The miner's profit $U(\{\TX^{i}\}_{i\in[n+1:n+m]},S,\tau)$ is defined as

    \begin{equation*}
        \sum_{i\in[|S|], \tau(i)\in [n+1:n+m]}
        \frac{x_{i-1}-x_i}{1-f\cdot \mathbbm{1}_{\{x_i > x_{i-1}\}}}\cdot p_x + \frac{y_{i-1}-y_i}{1-f\cdot \mathbbm{1}_{\{y_i > y_{i-1}\}}}\cdot p_y,
    \end{equation*}
    where $f\in[0,1)$ is the fraction of trading fees.
\end{definition}

Here, $\mathbbm{1}_{\{x_i > x_{i-1}\}}$ indicates that $\TX^{\tau(i)}$ is a $\Sell(\calX,\cdot)$ transaction and $\mathbbm{1}_{\{y_i > y_{i-1}\}}$ indicates that $\TX^{\tau(i)}$ is a $\Sell(\calY,\cdot)$ transaction. These two events will not happen simultaneously.

\subsection{Arbitrage-Free Interval}
In this subsection, we present a clean lemma that characterizes (what we call) arbitrage-free interval, which provides the first intuition behind the proofs later. It may also serve as the first step in other scenarios of decentralized exchanges when concerning the miner's strategies, \textit{e.g.}, optimal sandwich attacks of a miner who wants to collect money.

Before we state and prove the lemma, we first introduce a notation, which is also used in the subsequent sections. We use $L_x$ to denote the $x$ such that the marginal exchange rate $r(L_x)=\frac{1}{1-f}\frac{p_x}{p_y}$ and $R_x$ to denote the $x$ such that $r(R_x)=(1-f)\frac{p_x}{p_y}$.

\begin{lemma}\label{lemma: arbitrage free interval}
    Given the exogenous prices $p_x$ and $p_y$, and the current state $s^*=(x^*,y^*)$, miner's optimal profit is positive if and only if $x^*\not\in [L_x,R_x]$. Furthermore, when $x^*<L_x$, miner's optimal strategy is to execute $\Sell(\calX,(L_x-x^*)/(1-f))$; when $x^*>R_x$, miner's optimal strategy is to execute $\Sell(\calY,(F_y(R_x)-y^*)/(1-f))$.
\end{lemma}
\begin{proof}

    We first argue that it suffices for the miner to execute at most one transaction. This is because if miner executes two transactions with the same type (say $\Sell(\calX,q_1)$ and $\Sell(\calX,q_2)$), then it is equivalent to execute $\Sell(\calX,q_1+q_2)$; if miner executes two transactions with different types (say $\Sell(\calX,q_1)$ and $\Sell(\calY,q_2)$), then it is \emph{better} to replace them by one single transaction since miner can avoid additional cost of trading fees. 
    
    So next we consider the case where the miner executes one of its transactions $\TX$. Suppose that $\TX=\Sell(\calX,q)$, then miner's profit is 
    \[U(\calX,q)=\left(\int_{x^*}^{x^*+(1-f)q} r(x) dx\right)\cdot p_y-q\cdot p_x.\]

    We show below that when $x^*\geq L_x$, $U(\calX,q)\leq 0$ for all $q\geq 0$. 
\begin{equation*}
    \begin{aligned}
    U(\calX,q)&=\left(\int_{x^*}^{x^*+(1-f)q}r(x) dx\right)\cdot p_y-q\cdot p_x\\
    &\leq  r(x^*)(1-f)q\cdot p_y-q\cdot p_x\\
    &\leq  r(L_x)(1-f)q\cdot p_y-q\cdot p_x\\
    &= \frac{1}{1-f}\frac{p_x}{p_y}(1-f)q\cdot p_y-q\cdot p_x\\
    &= 0.
    \end{aligned}
\end{equation*}

    Symmetrically we can define $U(\calY,q)$ when miner executes $\Sell(\calY,q)$ and conclude that when $x^*\leq R_x$, $U(\calY,q)\leq 0$ for all $q\geq 0$. This finishes the proof that when $x^*\in[L_x, R_x]$, miners cannot obtain positive profits.

    Then we consider what is an optimal attack when $x^*\not \in[L_x, R_x]$. Suppose that $x^*<L_x$, then by previous argument, the miner should not execute $\Sell(\calY, \cdot)$ (as $x^*<L_x\leq R_x$). So let's focus on the case where the miner executes $\Sell(\calX,q)$. 

    Letting $x'=x^*+(1-f)q$, note that 
\begin{equation*}
    \begin{aligned}
    U(\calX,q)&=\left(\int_{x^*}^{x^*+(1-f)q} r(x) dx\right)\cdot p_y-q\cdot p_x\\
    &=\left(\int_{x^*}^{L_x} r(x)dx+\int_{L_x}^{x'}r(x) dx\right)\cdot p_y-q\cdot p_x, 
    \end{aligned}
\end{equation*}
where 
$\left(\int_{x^*}^{L_x}r(x) dx\right)\cdot p_y-(L_x-x^*)/(1-f)\cdot p_x$ is the profits that miner can get by executing $\Sell(\calX, (L_x-x^*)/(1-f))$ as states in the lemma. Next we show that \[g(x')=\left(\int_{L_x}^{x'}r(x) dx\right)\cdot p_y-(x'-L_x)/(1-f)\cdot p_x\leq 0\]
for all $x'$.

Note that \[g(x')=\left(F_y(L_x) - F_y(x')\right)\cdot p_y -(x'-L_x)/(1-f)\cdot p_x.\] So we have \[g'(x')=-F_y'(x')p_y-p_x/(1-f)=r(x')p_y-p_x/(1-f),\] which is a decreasing function as $r(x')$ is decreasing. Since $g'(L_x)=0$, we have the maximal value of $g$ is at $L_x$, which is 0.

This finishes the proof.
\end{proof}

\section{Strategies under Greedy Sequencing Rule}


In this section, we systemically analyze the strategic behaviors of the miners who \emph{follow} the greedy sequencing rule. 

We specifically focus on the case that the initial state $s_0=(x_0,y_0)$ satisfies $r(x_0)=p_x/p_y$. Note that this is without loss of generality in our context: On the one hand, when $f=0$, $L_x = R_x$ (\textit{i.e.}, the arbitrage-free interval becomes an arbitrage-free point). Supported by \Cref{lemma: arbitrage free interval}, if the current $\calX$ reserves are not $L_x$ ($R_x$), anyone can make money by a single 
arbitrage transaction, namely, by selling $\calX$ or $\calY$ to reach the arbitrage-free point. Thus, it is reasonable to think the last transaction ends up with the state $s_0=(x_0,y_0)$ satisfying $r(x_0)=p_x/p_y$, which is also the initial state of this attack; On the other hand, when $f>0$, we show that the NP-hardness holds even if $r(x_0)=p_x/p_y$, let alone the more general case. It is still interesting to consider the case $r(x_0)\neq p_x/p_y$, and we discuss it in the last section (\Cref{sec: discussion}).

In \Cref{sec: polytime when f=0}, we show a polynomial time algorithm to compute an optimal attack in the regime that the fraction of trading fee $f=0$. Interestingly, it will also \emph{benefit} the users if the miner follows such a strategy compared to truthfully following the greedy sequencing rule.

In contrast, \Cref{sec: NP-hard when f>0} shows that when the fraction of trading fee $f$ is any constant larger than 0 (say $f=0.3\%$ as being used in most Uniswap pools), it is NP-hard to find an optimal strategy. 

\subsection{Upper Bounds of Optimal Profits}
Our main results in this section (\Cref{theorem: polytime algorithm f=0} and \Cref{theorem: NP hard when f>0}) will be crucially based on the following lemma, which provides an upper bound of miner's optimal profit (using arbitrary strategy) under the greedy sequencing rule.

Before presenting the lemma, we first define the arbitragable profit for one transaction, inspired by \Cref{lemma: arbitrage free interval}.
\begin{definition}[Arbitragable Profit]
    Given an initial state $s_0=(x_0,y_0)$ and a user's transaction $\TX$, we define the arbitragable profit $\AP(s_0,\TX)$ as follows:
    \begin{itemize}
        \item If $\TX=\Sell(\calX,q)$, let $x'=\max\left\{x_0+(1-f)q,R_x\right\}$. Then $\AP(s_0,\TX)\coloneqq (x'-R_x)\cdot p_x - \left(F_y(R_x)-F_y(x')\right)/(1-f)\cdot p_y$;
        \item If $\TX=\Sell(\calY,q)$, let $x'=\min\left\{F_x(y_0+(1-f)q),L_x\right\}$. Then $\AP(s_0,\TX)\coloneqq \left(F_y(x')-F_y(L_x)\right)\cdot p_y - (L_x-x')/(1-f)\cdot p_x$.
    \end{itemize}
\end{definition}

\begin{figure}
    \centering
    \includegraphics[scale=0.5]{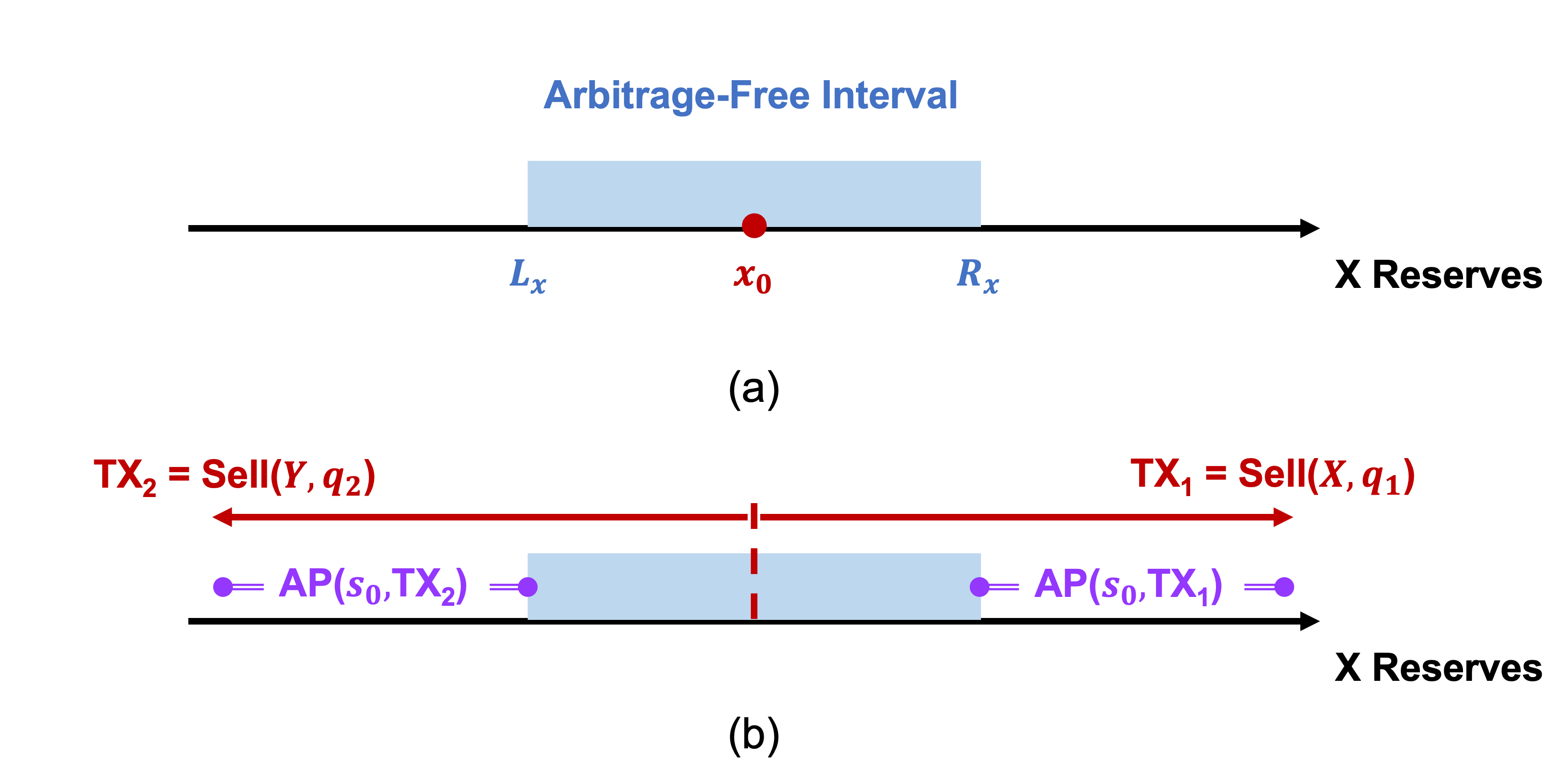}
    \caption{Illustration of Arbitrage-Free Interval and the intuition behind Arbitragable Profit.}
    \label{fig}
\end{figure}

\figureautorefname~\ref{fig} illustrates the intuition behind Arbitragable Profit.

The lemma below shows that the miner's optimal profit is upper-bounded by the sum arbitragable profits of all users' transactions.

\begin{lemma}\label{lemma: upper bound}
    Given an initial state $s_0=(x_0,y_0)$ with $r(x_0)=p_x/p_y$, a set of users' transactions $\{\TX^i\}_{i\in[n]}$, the miner's profit (using arbitrary strategy) under the greedy sequencing rule is upper bounded by $M(s_0,\{\TX^i\}_{i\in[n]})$, where 
    \[M(s_0,\{\TX^i\}_{i\in[n]})\coloneqq \sum_{i=1}^n \AP(s_0,\TX^i).\]
\end{lemma}
\begin{proof}
    Fix arbitrary sequence of (users' and miner's) transactions $(\TX^{\tau(1)},\cdots,\TX^{\tau(k)})$, where $\TX^{\tau(i)}$ is a user's transaction if $\tau(i)\in [n]$ and it is the miner's transaction otherwise. Let $s_i=(x_i,y_i)$ be the state after executing $\TX^{\tau(i)}$. 
    Without loss of generality, we assume that $\TX^{\tau(i)}=\Sell(\calX,\cdot)$ if and only if $x_{i-1}\leq x_0$ and $\TX^{\tau(i)}=\Sell(\calY,\cdot)$ if and only if $y_{i-1}\leq y_0$ for all $i\in\{2,\cdots,k\}$. To see it, suppose that for $k'<k$ we have $\TX^{\tau(i)}=\Sell(\calX,\cdot)$ and $x_{i-1}>x_0$ for all $i\in\{k'+1,\cdots,k\}$. Then by~\Cref{lemma: arbitrage free interval}, we know that miner's profit obtained from $\TX^{\tau(k'+1)},\cdots\TX^{\tau(k)}$ is at most 0 (and possibly negative). It means the miner can always choose not to execute these transactions and the profit is as good as before.

    We will inductively show that after executing the first $i$ transactions, the miner's profit $U_i\leq V_i\coloneqq \sum_{j\in[i],\tau(j)\in[n]}\AP(s_0,\TX^{\tau(j)})$. This will imply that after executing all $k$ transactions, miner's profit is upper bounded by $\sum_{i\in[n]} \AP(s_0,\TX^i)$.

    We define $\phi_i$ as follows:
    \begin{equation*}
    \phi_i=\left\{ \begin{array}{ll}
    (x_i-R_x)\cdot p_x+ \frac{F_y(x_i)-F_y(R_x)}{1-f}\cdot p_y, & x_i>R_x;\\
    \frac{x_i-L_x}{1-f}\cdot p_x+\left(F_y(x_i)-F_y(L_x)\right)\cdot p_y, & x_i<L_x;\\
    0, & x_i\in[L_x,R_x].
    \end{array}
    \right.
    \end{equation*}

    We will show that $(U_i+\phi_i)-(U_{i-1}+\phi_{i-1})\leq V_i-V_{i-1}=\AP(s_0,\TX^{\tau(i)})$ for all $i\in[k]$, which will imply our desired statement $U_i\leq V_i$ as $\phi_i\geq 0$ for all $i\in[k]$. (Here we define $\AP(s_0,\TX^{\tau(i)})=0$ if it is a miner's transaction.)

    The basis of the induction is trivial as $U_0+\phi_0=0$. For the induction step, let's consider arbitrary $i\in[k]$. 

    \textbf{Case 1:} $\TX^{\tau(i)}$ is a user's transaction. Then we have $U_i=U_{i-1}$. So it suffices for us to show $\phi_i-\phi_{i-1}\leq \AP(s_0,\TX^{\tau(i)})$. Suppose that $\TX^{\tau(i)}=\Sell(\calX,q)$. Then it must be the case $x_{i-1}<=x_0$ due to the greedy sequencing rule. 
    (The other case $\TX^{\tau(i)}=\Sell(\calY,q)$ will be symmetric.) If $x_i\leq x_0$, then we have that $\phi$ in fact didn't increase, which means $\phi_i-\phi_{i-1}\leq 0\leq \AP(s_0,\TX^{\tau(i)})$. If $x_i>x_0$, then since $x_{i-1}\leq x_0$, we have $x_i\leq \max\left\{x_0+(1-f)q,R_x\right\}$. So that $\phi_i\leq \AP(s_0,\TX^{\tau(i)})$, concluding the first case.
    
    \textbf{Case 2:} $\TX^{\tau(i)}$ is a miner's transaction. Then we have $V_i=V_{i-1}$. So it suffices for us to show $U_i-U_{i-1}+\phi_i-\phi_{i-1}\leq 0$. Suppose that $\TX^{\tau(i)}=\Sell(\calX,q)$, then it must be the case $x_{i-1}<=x_0$ due to the greedy sequencing rule. (Again, the other case $\TX^{\tau(i)}=\Sell(\calY,q)$ will be symmetric.) 
    
    If $x_{i-1}\leq x_i\leq L_x$, then we have in fact $U_i-U_{i-1}+\phi_i-\phi_{i-1}=0$ since $U_i-U_{i-1}=\phi_{i-1}-\phi_i=-(x_i-x_{i-1})/(1-f)\cdot p_x+\left(F_y(x_{i-1})-F_y(x_i)\right)\cdot p_y.$

    Now, let's consider the case $L_x\leq x_i$. To simplify the analysis, we consider an intermediate state $s'$ with $U'$ and $\phi'$. If $x_{i-1}\geq L_x$, then we just set $s'=s_{i-1}$ with $U'=U_{i-1}$ and $\phi'=\phi_{i-1}$. If $x_{i-1}<L_x$, we split $\TX^{\tau(i)}$ into two transactions: $\TX'=\Sell(\calX,(L_x-x_{i-1})/(1-f))$ and $\TX''=\Sell(\calX,(x_i-L_x)/(1-f))$, and we define $s'$, $U'$ and $\phi'$ as that after executing $\TX'$. 
    
    Note that we have $U'-U_{i-1}=\phi_{i-1}-\phi'$. So we only need to show $U_{i}-U'\leq \phi'-\phi_i$. Note that in fact $\phi'=0$.

    If $x_i\leq R_x$, then $\phi_i=\phi'=0$. In addition, by \Cref{lemma: arbitrage free interval}, we know that $U_i-U'\leq 0$. So we conclude $U_{i}-U'\leq \phi'-\phi_i$ as desired.

    The last possibility is that $x_i>R_x$, where we have \[\phi_i=(x_i-R_x)\cdot p_x+ \frac{F_y(x_i)-F_y(R_x)}{1-f}\cdot p_y.\]
    Moreover, by \Cref{lemma: arbitrage free interval}, we know that $U_i-U'\leq (R_x-x_i)/(1-f)\cdot p_x+\left(F_y(R_x)-F_y(x_i)\right)\cdot p_y < -\phi_i$.
    
    This finishes the proof.
\end{proof}

\subsection{Polynomial Time Algorithm When $f=0$}
\label{sec: polytime when f=0}
In this subsection, we show a polynomial time algorithm to find an optimal strategy for the miner when $f=0$. Interestingly, when adopting our algorithm, users will also benefit in the following sense: all users' transactions will be executed (\textit{a.k.a} they will be included in the block), and every user gets at least as good as if their transaction was the only one in the block. The latter is generally not true if the miner truthfully follows the greedy sequencing rule.

\begin{algorithm}[!t]
	\caption{Algorithm for optimal strategy when $f=0$}
	\label{alg: f=0}
	\KwIn{An initial state $s_0=(x_0,y_0)$, and a set of users' transactions $\{\TX^i\}_{i\in[n]}$.
	}
	\KwOut{An optimal strategy for miner to obtain $M(s_0,\{\TX^i\}_{i\in[n]})$ profits, which is the best possible.
	}
	
	\BlankLine
	
	Sort these $n$ transactions in any order $\tau:[n]\rightarrow[n]$.
 
    \For{each $i$ from $1$ to $n$}{
    Execute user's transaction $\TX^{\tau(i)}$.

    \If{$\TX^{\tau(i)}=\Sell(\calX,q)$}{\label{line: if}
    Execute a transaction $\Sell(\calY,y_0-F_y(x_0+q))$.\label{line: tx y}
    }
    \If{$\TX^{\tau(i)}=\Sell(\calY,q)$}{Execute a transaction $\Sell(\calX,x_0-F_x(y_0+q))$.\label{line: tx x}
    }\label{line: end if}
    }
    
\end{algorithm}

\begin{theorem}\label{theorem: polytime algorithm f=0}
    When the fraction of trading fee $f=0$, Algorithm~\ref{alg: f=0} finds an optimal strategy under the greedy sequencing rule in polynomial time, and the optimal profit is equal to the upper bound $M(s_0,\{\TX^i\}_{i\in[n]})$.
\end{theorem}
\begin{remark}
    Before going into details of the proof, we note that our algorithm can obtain the optimal profit $M(s_0,\{\TX^i\}_{i\in[n]})$ under arbitrary order of users' transactions $\{\TX^i\}_{i\in[n]}$. So it still works even if there are some constraints on the execution order of certain transactions (\textit{e.g.}, a user may create two transactions $\{\TX^1, \TX^2\}$ and specify that $\TX^1$ must be executed before $\TX^2$).
\end{remark}

\begin{proof}[Proof of \Cref{theorem: polytime algorithm f=0}]
    We first show that the sequence given by Algorithm~\ref{alg: f=0} satisfies the greedy sequencing rule. Note that after executing each user's transaction $\TX^{\tau(i)}$, we always execute a miner's transaction with the opposite direction, shown between line~\ref{line: if} and \ref{line: end if}. Besides, at the end of $i$-th iteration, we have the state $s_{2i}=s_0$ (we use $2i$ because we execute two transactions in each iteration). So our sequence satisfies the greedy sequencing rule. Furthermore, during the $i$-th iteration, we obtain exactly $\AP(s_0,\TX^{\tau(i)})$ profits by executing the transaction on line~\ref{line: tx y} or \ref{line: tx x}. Then the optimality follows from the same upper bound provided by \Cref{lemma: upper bound}.
\end{proof}

Now we turn to the \emph{positive} effects on users when a miner launches an optimal strategy given by Algorithm~\ref{alg: f=0}. We summarize them as the following two corollaries and omit the proofs as they are relatively straightforward from the proof of \Cref{theorem: polytime algorithm f=0}.
\begin{corollary}\label{corollary: liveness}
    When a miner launches an optimal strategy given by Algorithm~\ref{alg: f=0}, all users' transactions $\{\TX^i\}_{i\in[n]}$ will be executed.
\end{corollary}
\begin{corollary}\label{corollary: good profit}
    When a miner launches an optimal strategy given by Algorithm~\ref{alg: f=0}, each user's profit is as good as if their transaction was the only transaction in the block.
\end{corollary}
As shown in \Cref{theorem: polytime algorithm f=0}, \Cref{corollary: liveness}, \Cref{corollary: good profit}, both miner and users are satisfied when miner adopts our Algorithm~\ref{alg: f=0}.

\subsection{NP-hardness When $f>0$}
\label{sec: NP-hard when f>0}
In this subsection, we show the computational hardness of finding an optimal strategy when the fraction of trading fees is any constant larger than 0 (say $f=0.3\%$). 

We will mainly focus on the proof of the NP-completeness of the following decision problem, then \Cref{theorem: NP hard when f>0} will follow directly.
\begin{theorem}\label{theorem: NP hard of decision}
    Let $f\in(0,1)$ be any universal constant. It is NP-complete to decide if there is a strategy that can obtain profits $M(s_0,\{\TX^i\}_{i\in[n]})$ for any initial state $s_0=(x_0,y_0)$ and users' transactions $\{\TX^i\}_{i\in[n]}$.
\end{theorem}
\begin{proof}
    The NP-membership is easy. Given any strategy, we can efficiently simulate the execution of the sequence of transactions and check if the final profit is $M(s_0,\{\TX^i\}_{i\in[n]})$ or not.

    For the NP-hardness, we reduce the Partition problem to our problem. Recall that the instance of the partition problem contains $n$ positive integers and ask if it can be partitioned into two subsets $S_1$ and $S_2$ such that the sum of numbers in $S_1$ equals that in $S_2$.

    Suppose we are given arbitrary $n$ positive integers $\{a_1,\cdots,a_n\}$. Let $t$ be half of the sum of these integers, \textit{i.e.}, $\frac{1}{2}\sum_{i=1}^{n} a_i$. Without loss of generality, we assume that $a_i\leq t$ for all $i\in[n]$ otherwise the answer to the decision problem will directly be ``no''.
    
    We first construct a CFMM $A$ and initial state $s_0$. Concretely, we can consider the constant curve of $A$ as $F(x,y):xy=k$, and our goal is to choose parameters such that $x_0-L_x=(1-f)t$. Precisely, we know that $L_x=\sqrt{1-f}x_0$, since $r(L_x)=\frac{1}{1-f}r(x_0)$. This means $x_0-L_x=(1-\sqrt{1-f})x_0$. So choosing $x_0=\frac{1-f}{1-\sqrt{1-f}}t$ would suffice.

    Next, we construct users' transactions. For each integer $a_i$, we construct $\TX^{i}=\Sell(\calX,a_i)$. Clearly, we have $\AP(s_0,\TX^i)=0$ as $(1-f)a_i \leq (1-f)t = x_0-L_x \leq R_x-x_0$. Then we construct two $\Sell(\calY,\cdot)$ transactions. Precisely, we construct $\TX^{n+1}=\TX^{n+2}=\Sell(\calY, q^*)$ where $q^*$ is large enough such that $F_x(y_0 + (1-f)q^*)<L_x$. Then we know $\AP(s_0,\TX^{n+1})=\AP(s_0,\TX^{n+2})>0$. This finishes the construction. And we know $M(s_0,\{\TX^i\}_{i\in[n+2]})=2\AP(s_0,\TX^{n+1})$.

    Finally, we argue that there exists a strategy obtaining profits $M(s_0,\{\TX^i\}_{i\in[n+2]})$ if and only if there exists a subset $S\subseteq [n]$ such that the sum of the numbers in $S$ equal $t$. And this will conclude the theorem.

    One direction is easy: if there exists $S\subseteq [n]$ such that the sum of the numbers in $S$ equal $t$, then we execute transactions as follows: \begin{enumerate}
        \item Execute user's transaction $\TX^{n+1}$; 
        Execute miner's transaction $\Sell(\calX, \frac{L_x-F_x\left(y_0 + (1-f)q^*\right)}{1-f})$;\vspace{-0.2cm}
        \item Execute $\TX^i$ for all $i\in S$; \vspace{-0.2cm}
        \item Repeat item (1) except replacing $\TX^{n+1}$ by $\TX^{n+2}$.
    \end{enumerate} 
    It is easy to verify that this sequence satisfies the greedy sequencing rule, and the miner can obtain $M(s_0,\{\TX^i\}_{i\in[n+2]})$.

    For the other direction, we show that the sequence of transactions constructed above is essentially the only way to obtain $M(s_0,\{\TX^i\}_{i\in[n+2]})$. So a miner can obtain $M(s_0,\{\TX^i\}_{i\in[n+2]})$ only if the answer to the given Partition problem is ``yes''.

    We adopt a proof scheme similar to that of \Cref{lemma: upper bound}. Fix a sequence of (users' and miner's) transactions $(\TX^{\tau(1)},\cdots,\TX^{\tau(k)})$ such that miner's profit $U=2\AP(s_0,\TX^{n+1})$. Recall that in the proof of \Cref{lemma: upper bound}, we defined $\phi_i$ and showed $U_i+\phi_i-(U_{i-1}+\phi_{i-1})\leq \AP(s_0,\TX^{i})$ for all $i\in[k]$. Since $U_k=2\AP(s_0,\TX^{n+1})$ at the end, it must be the case $U_i+\phi_i=V_i$ for all $i\in[k]$ and $\phi_k=0$. As a result, the sequence of transactions must satisfy that
    \begin{itemize}
        \item The miner does not lose profit for any transaction; otherwise the loss of the profit is strictly larger than the gain of the $\phi$ function, and this will result in $U_i+\phi_i<V_i$ for some $i$.
        \item There are $i_1\neq i_2\in[k]$ such that $\phi_{i_1}=\phi_{i_2}=\AP(s_0,\TX^{n+1})$. This means when execute $\TX^{n+1}$ and $\TX^{n+2}$, the corresponding state must be $(x_0,y_0)$.
    \end{itemize}

    To achieve both items simultaneously, it must be $(x_{i_1-1},y_{i_1-1})=(x_0,y_0)$ and $\TX^{n+1}$ is executed as $\TX^{\tau(i_1)}$. To get the first $\AP(s_0,\TX^{n+1})$ profit, miner executes $\Sell(\calX, \frac{L_x-F_x\left(y_0 + (1-f)q^*\right)}{1-f})$ in the $(i_1+1)$-th iteration. To make sure that $(x_{i_2-1},y_{i_2-1})=(x_0,y_0)$ (and $\TX^{n+2}$ is executed as $\TX^{\tau(i_2)}$) while the miner does not loss any profit in this process, we must use users' transactions to change the state from $x_{i_1}=L_x$ to $x_{i_2-1}=x_0$, which means we need a subset $S$ of users' transactions such that the sum of numbers in $S$ is exactly $t$.

    This finishes the proof.
\end{proof}

\Cref{theorem: NP hard when f>0} follows directly by simulating any algorithm that computes an optimal strategy and calculates the profits to solve the decision problem.

\begin{theorem}\label{theorem: NP hard when f>0}
    Let $f\in(0,1)$ be any universal constant. It is NP-hard to compute the strategy that can obtain the optimal profits.
\end{theorem}

\section{Discussion and Open Problems}
\label{sec: discussion}

\noindent\textbf{Refined Sequencing Rule.} Our first question is related to mechanism design, motivated by a revisit of our polynomial time algorithm when $f=0$. Recall that our algorithm can always obtain the upper bound profits, even if the miner is asked to follow the greedy sequencing rule such that the sequence is additionally under a descending order. Thus, we would like to ask if there is some sequencing rule (that is computationally efficient and verifiable) that can further mitigate the miner's incentive to manipulate. We propose the following way to build a theoretical foundation when considering real-world applications. We could consider the case where users' transactions are drawn from a certain distribution $\calD$ (witnessed by real-world DeFi scenarios), and show that under the refined greedy sequencing rule, miners cannot obtain large profits with high probability. We leave it as a promising open question. 

\vspace{1.5mm}

\noindent\textbf{Approximation Algorithm for Miners.}
It is also worth to study about approximation algorithm design for miners. 
Our NP-hardness rules out the possibility for a miner to have a polynomial time algorithm for an optimal strategy (assuming P$\neq$NP). 
However, it remains possible to design a polynomial time algorithm with a good approximation guarantee. 
This strategy exploration allows miners to develop efficient algorithms that can yield sufficient MEV close to the optimal strategy.
As the optimal MEV problem shares a similar spirit with the Knapsack problem, one promising direction is to apply the classic approximation algorithms to our setting.

\vspace{1.5mm}
\noindent\textbf{User's Strategies.}
The third question is about strategic analysis from the perspective of users. In this work, we systematically studied the optimal strategies of miners. We also note that there is fruitful space for a user to adopt strategies. For example, a user who wants to sell a large amount of $\calX$ tokens may have an incentive to split it into several smaller transactions, and this may lead them to a higher profit under the greedy sequencing rule. Generally speaking, we wonder what is an optimal strategy for a user under certain sequencing rules. Different from the miner's incentive, multiple users are making decisions simultaneously, which forms a multi-agent system. One step further than one user's optimization, we ask what the equilibrium is when all users behave strategically.
The game theory problem between users and miners under specific sequencing rules is also an intriguing question.

\vspace{1.5mm}
\noindent\textbf{Other Scenarios where MEV Makes Everyone Happy.}
Finally, recall our exciting journey about the positive effects of MEV: when a miner attracts MEV (optimally), users are also benefited in a reasonable sense (\Cref{corollary: liveness} and \Cref{corollary: good profit}). 
The intuition behind this phenomenon is that although the existence of MEV incentivizes miners to engage in attacking behaviors when a good sequencing rule can restrict miners' actions and prevent them from affecting users' profits, the presence of MEV itself can benefit users.
In this case, MEV not only does not harm users but can expedite the execution of user transactions as miners have the motivation to execute more transactions (to obtain MEV).
We expect and are eager to know a wider range of scenarios where the same conceptual result also holds. We leave this as the most important future work.


\bibliographystyle{IEEEtran}

\end{document}